\documentclass[letterpaper]{article} 
\usepackage{aaai2026}  
\usepackage{times}  
\usepackage{helvet}  
\usepackage{courier}  
\usepackage[hyphens]{url}  
\usepackage{graphicx} 
\usepackage{natbib}  
\usepackage{caption} 
\usepackage{algorithm}
\usepackage{algorithmic}

\usepackage{amsfonts,amsmath,amssymb,amsthm}

\newtheorem{theorem}{Theorem}
\newtheorem{corollary}{Corollary}
\newtheorem{assumption}{Assumption}

\theoremstyle{definition}
\newtheorem{definition}{Definition}
\DeclareMathOperator{\E}{\mathbb{E}}

\newcommand{\optAct}[2]{A^{#1,*}(#2)}

\newcommand{\optActOuter}[2]{\bar{A}^{#1,*}(#2)}

\newcommand{\optValueVec}[1]{V^{\ast}_{\Pi_i}}
\newcommand{\optValue}[2]{\optValueVec{#1}(#2)}
\newcommand{\approxValueVec}[1]{\tilde{V}_{\Pi_{#1}}}

\newcommand{\overValueVec}[1]{\overline{V}_{\Pi_{#1}}}
\newcommand{\underValueVec}[1]{\underline{V}_{\Pi_{#1}}}
\newcommand{\overValue}[2]{\overValueVec{#1}(#2)}
\newcommand{\underValue}[2]{\underValueVec{#1}(#2)}

\usepackage{newfloat}
\usepackage{listings}
\DeclareCaptionStyle{ruled}{labelfont=normalfont,labelsep=colon,strut=off} 
\floatstyle{ruled}
\newfloat{listing}{tb}{lst}{}
\floatname{listing}{Listing}

\usepackage{pgfplotstable}
\usepackage{pgfplots}
\pgfplotsset{compat=1.18}
\title{Social Laws for Multi-agent Coordination in Stochastic Environments}
\author{
    Rolando Fernandez\textsuperscript{\rm 1},
    Caleb Probine\textsuperscript{\rm 1},
    Tyler Lee\textsuperscript{\rm 1},
    Jeffrey Chen\textsuperscript{\rm 1},
    Erez Karpas\textsuperscript{\rm 1}\textsuperscript{\rm 2},
    Muhammad Arrasy Rahman\textsuperscript{\rm 1},
    Peter Stone\textsuperscript{\rm 1},    
    Ufuk Topcu
    \textsuperscript{\rm 1}             
}
\affiliations{
    \textsuperscript{\rm 1}The University of Texas at Austin\\
    \textsuperscript{\rm 2}Technion Israel Institute of Technology\\
    rfernandez@utexas.edu, cprobine@utexas.edu, tylerblee9@gmail.com, jeffrey9@utexas.edu, karpas@cs.utexas.edu, arrasy@cs.utexas.edu, pstone@cs.utexas.edu, utopcu@utexas.edu
}

\begin{document}

\maketitle

\begin{abstract}
In multi-agent environments, coordinating agents to prevent interference and ensure robust individual performance is a critical challenge. Previous research on social laws for multi-agent systems has primarily focused on deterministic, goal-based settings. This paper extends the concept of social laws to stochastic, reward-based environments, proposing a formalism for defining and verifying their robustness under various conditions. We introduce the notion of $\alpha$-robustness, a measure of the guaranteed utility each agent retains while pursuing its optimal single agent policy, assuming all agents obey the social law. We then present an approach for robustness verification of social laws in stochastic settings, based on a reduction to solving a series of Markov decision processes. Empirical evaluations on toy environments illustrate the potential of our framework.
\end{abstract}




\section{Introduction}

An agent operating in a multi-agent environment must reason not just about its own possible actions, but also about the possible actions of the other agents in the environment. Even in deterministic settings, this can make the agent's planning problem exponentially harder in the number of agents \cite{brafman2008one}. In this paper, we consider the specific case in which the agents are neither adversarial nor cooperative, but instead must coordinate among themselves to avoid interfering with one another. 

Coordination in multi-agent systems has been a longstanding area of research across fields such as artificial intelligence, robotics, and game theory. In environments where multiple agents operate simultaneously, each with distinct objectives and capabilities, ensuring robust performance becomes a fundamental challenge. A promising approach to mitigate inter-agent interference and promote harmonious operations is the use of social laws. These laws act as constraints on the actions of agents, limiting behaviors that could disrupt each other's ability to achieve individual goals. Previous work has demonstrated the efficacy of social laws in deterministic, goal-driven environments \citep{Shoham1992, Shoham1995}. These approaches highlighted how social laws provide decentralized coordination, avoiding complex negotiations and ensuring predictable behavior.

Despite the success of social laws in deterministic environments, their application to stochastic multi-agent systems remains largely unexplored. In stochastic settings, actions have probabilistic outcomes, introducing randomness that complicates planning and coordination. Moreover, previous work on social laws in deterministic environments assumed that each agent had a specified goal. In contrast, we aim to support a more expressive reward-driven model in stochastic settings. Additionally, the concurrent nature of stochastic games introduces opportunities for agents to act simultaneously, enabling greater flexibility in modeling real-world multi-agent systems than in previous work on planning-based social laws, which relied on {\em waitfor} preconditions \cite{Karpas2017}.

To address the gap between deterministic, goal-driven models and stochastic, reward-driven models, we propose a new formalism that extends social laws to stochastic multi-agent environments. At the core of this framework is the concept of $\alpha$-robustness, which quantifies how effectively agents retain a guaranteed portion of the utility they could expect to collect in a simplified world (in which they can perform single agent planning),
while operating in the real environment under the assumption that all agents obey the social law.
By reframing social law robustness within the mathematical structure of stochastic games, we provide an innovative tool for designing coordination mechanisms that prevent interference while maximizing agents' utilities, allowing us to discuss notions of fairness and efficiency.


Our contributions are threefold. First, we introduce $\alpha$-robustness, a new metric for assessing the robustness of multi-agent environments that extends traditional robustness definitions to reward-driven stochastic systems. 
Second, we propose a computational technique for robustness verification, 
which relies on optimally solving a series of smaller Markov decision processes (MDPs).
Finally, we 
show how $\alpha$-robustness and the guaranteed utility from the single agent projections—which represent the MDP of an agent acting under the assumption that other agents are performing some default policy (i.e., noop)—provide a lower bound on the utility each agent can guarantee, allowing a designer to make principled decisions about which social law to adopt.

This paper not only bridges the gap between deterministic social laws and stochastic settings but also sets the foundation for 
coordination challenges in uncertain multi-agent domains.


\section{Related Work}
Social laws have long been recognized as an effective mechanism for achieving coordination among autonomous agents in multi-agent systems. The concept of social laws was introduced by \citet{Shoham1992, Shoham1995}, defining them as constraints on agents’ actions that restrict harmful behaviors and ensure individual goals can be achieved without interference in deterministic, goal-driven environments. These pioneering efforts highlighted the benefits of social laws for decentralized coordination, the avoidance of complex agent negotiations, and the promotion of predictable behavior in multi-agent systems. \citet{Tennenholtz1995} extended this idea by exploring the applications of social laws to mobile robot systems, providing a framework for using traffic-like rules to ensure cooperative navigation and prevent conflicts.

Building on these foundations, \citet{Karpas2017} extended social laws to deterministic planning systems, focusing on goal-driven environments in which agents act independently but require coordination to prevent interference. Their framework introduced the concept of social law robustness using MA-STRIPS, a multi-agent planning representation that formalizes rational and adversarial robustness. They proposed an efficient compilation approach to classical planning that allows automated verification of the robustness of social laws under various conditions. In a subsequent extension, \citet{Nir2020} presented a method for verifying social laws in numerical settings, accommodating multi-agent systems in which numerical constraints influence agents’ behavior. Their results demonstrated the feasibility of refining social laws within more complex environments, paving the way for advanced models of social law synthesis.

Recent advances have sought to address the dynamic nature of real-world applications. \citet{Tuysov2020} introduced a framework to verify the robustness of social laws for reactive agents, offering a novel solution to scenarios in which agents are allowed to replan during execution failures. Similarly, \citet{Nir2019} extended social laws to multi-agent coordination in continuous time multi-robot systems. The authors presented methodologies to ensure robust coordination in continuous domains, addressing challenges arising from agent concurrency. These contributions expand the scope of social laws, enabling their application beyond the deterministic assumptions of classical planning environments. However, the notion of concurrency in the temporal planning setting differs from that in the stochastic game formulation in this paper.

Social norms provide a complementary perspective to social laws in multi-agent systems, focusing on the implicit rules or conventions that govern agent behavior. \citet{Agotnes2010} emphasized the role of normative systems, describing them as sets of constraints imposed on agents to ensure orderly and predictable behavior. These norms differ from social laws in that they are less explicit and rely on an assumed level of compliance among agents. The authors presented a formalization of normative systems and their use in preventing conflict and facilitating coordination in complex multi-agent environments. Similarly, \citet{Agotnes2012} expanded on normative frameworks by introducing conservative social norms that impose minimum constraints while ensuring that agents can achieve their individual goals without interference. This work provided insights into the design of systems in which agents voluntarily adhere to established norms to maintain collective stability.

A different notion of robustness was explored in the work on social norms \cite{Wooldridge2009} under the title of robust normative systems. They proposed strategies for modeling and enforcing norms in scenarios where agent compliance could not be guaranteed. They defined k-robustness, which quantifies the ability of normative systems to function effectively even when a certain number of agents fail to comply with the prescribed norms. This notion is particularly relevant to real-world systems with varying levels of agent compliance, such as autonomous vehicles or distributed robotic teams. By integrating norms with decision-making processes, the authors illustrated how a system can preserve its overall functionality despite partial non-compliance. Furthermore, \citet{Agotnes2012} discussed scenarios where social norms yield benefits beyond simple utility maximization, highlighting their importance for maintaining system-wide harmony and safety.

While social norms provide a flexible and adaptive framework, they differ fundamentally from social laws, which impose explicit restrictions on agents' behavior without accounting for noncompliance. Social laws are typically applied in systems that assume deterministic compliance among agents, as demonstrated by \citet{Shoham1992}, \citet{Shoham1995}, and \citet{Karpas2017}. However, both approaches address similar challenges in coordination and robustness, making them highly relevant to the synthesis of rules in multi-agent settings. 


This paper builds upon these foundational contributions by introducing $\alpha$-robustness, a novel notion of robustness of 
social laws in stochastic environments. Unlike prior work that focuses on deterministic or numerical settings, we address the challenges posed by stochastic transitions and reward-driven behavior in uncertain domains. 




Other work has addressed similar challenges in stochastic settings. 
Symbolic verification techniques were used in stochastic games
\cite{Kwiatkowska2022}, but they were limited to turn-based games, while this paper allows for concurrent action.
Robust Markov Games \cite{pmlr-v235-mcmahan24a,10.5555/3692070.3693899} are also similar, but have been used to address reinforcement learning when there is uncertainty about the model itself, while in this paper, we assume the model is known. Finally, other work shows that a new agent can learn an existing social norm \cite{10.5555/3635637.3663011}, further emphasizing the relevance of the problem we address in this paper: social law robustness verification.

\section{Background}

We now provide some background necessary to understand our problem formulation.
Stochastic games \cite{Shapley1953} generalize both Markov decision processes (MDPs) and strategic games, providing a framework for modeling decision-making in multi-agent environments in which action outcomes are probabilistic. A stochastic game is defined as a tuple $\Pi = \langle N, S, \{A^i\}_{i=1}^{|N|}, I, \{R_i\}_{i=1}^{|N|}, P, \gamma \rangle$, where: 
    \begin{itemize}
        \item $N$ is a set of agents
        \item $S$ is the set of possible world states
        \item $A^i$ is the set of action labels available to agent $i \in N$
        \item $I$ is a distribution over the initial state
        \item $R^i : S \times A^1 \times \ldots \times A^{|N|} \rightarrow \mathbb{R}$ is the reward function for agent $i \in N$
        \item $P : S \times A^1 \times \ldots \times A^{|N|} \times S \rightarrow [0,1]$ is the transition function, and
        \item $\gamma$ is the discount factor.
        
    \end{itemize}


In stochastic games, each agent selects an action, and their collective choices determine probabilistic transitions between states according to the function $P$. The reward functions $R^i$ specify the utility each agent gains as a function of the current state and the action taken. The concurrent nature of stochastic games means that all agents act simultaneously, in contrast to sequential decision-making frameworks such as MDPs, where a single agent executes actions in isolation.

Stochastic games are often studied in game theory, where solutions such as Nash equilibria describe stable joint strategies of agents \citep{FilarVrieze1997}. These equilibria ensure that no agent can unilaterally deviate from its strategy to achieve a higher expected reward. In practice, solving stochastic games is challenging due to their inherent complexity, as the interaction among multiple agents creates large state and action spaces. Techniques such as mathematical programming \citep{Daskalakis2009} and reinforcement learning \citep{HuWellman1998} have been developed to address this computational challenge.

While much of the focus in stochastic game research has been on modeling cooperation and competition among agents, the unique challenge of coordination without cooperation has received comparatively less attention. Coordination in non-cooperative stochastic games requires mechanisms to prevent interference while still enabling agents to pursue individual reward-maximizing policies. This paper builds on the foundation of stochastic games and introduces social laws as a mechanism to enforce coordination, ensuring that agents can maintain robust performance without the need for collaboration.

\section{Formal Model}

We begin by describing our formal model for social laws in stochastic environments. 

\begin{definition}[Multi-agent Environment]
    The environment $\Pi$ is a stochastic game, given by $\Pi = \langle N, S, \{A^i\}_{i=1}^{|N|}, I, \{R^i\}_{i=1}^{|N|}, P, \gamma \rangle$.
        

    On top of the standard definition of a stochastic game, we assume each agent has a {\em default policy} $\pi_d^i : S \rightarrow A_i$, representing its default behavior. In many cases, we can assume the default policy consists of applying a $\mbox{noop}$ action at each state, but our formalism applies to any deterministic default policy. 
\end{definition}

\begin{definition}[Single Agent Projection]    
We can now define the {\em single agent projection} for agent $i$, denoted $\Pi_i$ to be $\Pi_i = \langle S, A^i, I, R^i, P|_i, \gamma \rangle$ -- an MDP with
\begin{itemize}
    \item State space $S$
    \item Actions $A^i$
    \item Initial distribution $I$
    \item Reward function $R^i$
    \item Transition function $P|_i(s, a, s')$, induced by all agents except $i$ following their default policy: 
    $P|_i(s, a, s') := P(s, \langle \pi_d^1(s), \ldots, \pi_d^{i-1}(s),  a, \pi_d^{i+1}(s), \ldots, \pi_d^{n}(s) \rangle, s')$,  and
    \item Discount factor $\gamma$.
\end{itemize} 
\end{definition}

The single agent projection $\Pi_i$ is an MDP in which agent $i$ can plan, assuming all other agents execute their default policy. The main benefit of the single agent projection is computational; agents need not consider all possible interactions with other agents. Still, they can plan as if they were alone in the world, assuming other agents act predictably. 
This benefit is evident by the size of the action space (branching factor) of $\Pi_i$ being $|A^i|$, which is exponentially lower than the size of the action space of the original stochastic game: 
$|A^1|\cdot \ldots \cdot |A^{|N|}|$.
We denote the optimal value function of $\Pi_i$ by $V^\ast_{\Pi_i}$.

Of course, planning, assuming you are alone in the world, can yield incredibly bad results when the basic assumption about the behavior of the other agents is wrong. For example, autonomous vehicles could collide if each assumes the others do not move. However, we now define a suitable notion of {\em robustness} for stochastic multi-agent environments, such that when the environment is robust, each agent can plan as if it were alone in the world (i.e., in its single agent projection).  Under the assumption that all agents follow an optimal policy of their own single agent projection, we can provide a lower bound on the utility it will collect in the real environment.
We begin by defining the guaranteed utility.

\begin{definition}[Guaranteed Utility]
    The {\em guaranteed utility} for agent $i$, $E_i$, is the expected utility (discounted total sum of rewards) of an optimal policy in the single agent projection $\Pi_i$ from the initial distribution $I$, that is,
    $E_i := \E_{s_0 \sim I} V^\ast_{\Pi_i}(s_0)$.
\end{definition}

This value is the utility agent $i$ can achieve by itself, assuming the other agents do not ``interfere'' and do nothing (follow the default policy). Note that this is not really guaranteed in the strict sense of agent $i$ being able to achieve this value regardless of the actions of the other agents. Rather, it is the benchmark against which we measure the robustness of the multi-agent environment, as described next.



\begin{definition}[$\alpha$-robustness]
\label{d:robustness}
A multi-agent environment $\Pi$ is {\em $\alpha$-robust for agent $i$} (for $\alpha \in (0,1]$) iff the utility for agent $i$ from the joint policy $\pi^\ast_1 \times \ldots \times \pi^\ast_n$ is at least $\alpha \cdot E_i$, assuming $\pi^\ast_j$ is an optimal policy of the single agent projection $\Pi_j$, for every agent $j$. 
We say $\Pi$ is {\em $\alpha$-robust} if it is $\alpha$-robust for every agent.
\end{definition}

We remark that if the reward can be positive or negative, $\alpha$-robustness may be negative. Thus, we assume that the rewards in the multi-agent environment are all positive (although a similar notion of $\alpha$-robustness could be defined for all negative rewards, i.e., costs).

Coming back to our discussion about planning in the single agent projection, having an $\alpha$-robust multi-agent environment guarantees that each agent can optimally solve its single agent projection (which, as mentioned above, is much easier than optimally solving the full stochastic game), and can be sure that the utility it will collect in the real environment is at least an $\alpha$ fraction of its guaranteed utility, $E_i$. 
The proof of the following corollary is immediate from Definition \ref{d:robustness}:
\begin{corollary}
\label{c:guarantee}
Let $\Pi$ be a multi-agent environment, let $i$ be an agent, and let $E_i$ be the guaranteed utility of $i$ in $\Pi$. Assume $\Pi$ is $\alpha$-robust for agent $i$, and assume every agent $j$ follows an optimal policy of its single agent projection $\Pi_j$.
Then agent $i$ will collect a utility of at least $\alpha \cdot E_i$.
\end{corollary}

In words, 
if $\Pi$ is $\alpha$-robust, then every agent $i$ can guarantee achieving a utility of at least $\alpha \cdot E_i$ by following an optimal policy for its single agent projection, assuming all other agents also do the same. 
In a sense, the loss of a factor of $1-\alpha$ is the price we pay for solving the much easier single agent planning problem.
Having defined the notion of $\alpha$-robustness, we can define several computational problems:
\begin{itemize}
    \item Given a multi-agent environment $\Pi$ and $\alpha$, is $\Pi$ $\alpha$-robust?
    \item Given a multi-agent environment $\Pi$, find the highest $\alpha$ for which it is $\alpha$-robust.  
\end{itemize}

\subsection{Social Laws}

As in previous work on social laws, we can define a social law $l$ as a transformation of $\Pi$. Applying $l$ to $\Pi$ yields a new stochastic game $\Pi^l$, which represents the same multi-agent environment -- except with restrictions placed on the agents. Denote the robustness level of $\Pi^l$ by $\alpha_l$, and the guaranteed utility of agent $i$ in $\Pi^l$ by $E_i^l$.
From Corollary \ref{c:guarantee}, the utility agent $i$ can guarantee it will achieve is $\alpha_l \cdot E_i^l$, under the assumption that all agents follow an optimal policy of the single agent projection and respect the social law.

Let us now consider an alternative in which there is no social law. In this case, agent $i$ can only guarantee anything by assuming the worst-case behavior of the other agents. Thus, agent $i$ must solve a min-max optimization problem involving the actions of all agents, making this computationally expensive. Even ignoring the computational cost of this approach, the following example demonstrates that, in some cases, social laws are required to guarantee anything beyond the worst possible outcome.

\paragraph{Example 1:} Consider a simple 2x2 gridworld, illustrated in Figure \ref{f:example}, in which the blue agent starts at the northwest, and collects a reward of +1 at the southeast, while the red agent starts at the southeast and collects a reward of +1 at the northwest. Each agent can move north, south, east, or west. If both agents end up at the same location, they enter a dead-end collision state, yielding a reward of 0. All other states also have a reward of 0. For the sake of brevity, assume all actions are deterministic.

Without a social law, there is no way to guarantee avoiding collisions, and thus, the min-max value each agent can guarantee is 0. On the other hand, by introducing a social law that allows agents to move only clockwise, we can ensure that each agent reaches its destination and receives a reward of +1.

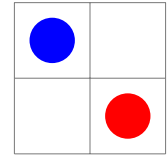
\begin{figure}
    \centering
\begin{tikzpicture}
    \draw[step=1cm, gray, thin] (0,0) grid (2,2);
    
    \fill[blue] (0.5, 1.5) circle (0.3cm);
    
    \fill[red] (1.5, 0.5) circle (0.3cm);
\end{tikzpicture}
    \caption{Illustration of Example Gridworld Environment}
    \label{f:example}
\end{figure}

\subsection{Comparing Social Laws}

When choosing which social law to institute, a designer must sometimes weigh competing social laws. For example, in the gridworld from Example 1, one could also only allow counter-clockwise motions. In this case, both social laws yield the same utility, but this is not always the case.

Given the same multi-agent environment $\Pi$, we could also compare the utility of two different social laws, $l_1$ and $l_2$, using Corollary \ref{c:guarantee}. Assuming $\Pi^{l_1}$ is $\alpha_1$-robust and $\Pi^{l_2}$ is $\alpha_2$-robust, and denoting the guaranteed utility for agent $i$ in $\Pi^{l_1}$ (respectively, $\Pi^{l_2}$) by $E_i^1$ (respectively, $E_i^2$), then agent $i$ can guarantee achieving $\alpha_1 \cdot E_i^1$ under $l_1$ and $\alpha_2 \cdot E_i^2$ under $l_2$. 

The following example illustrates the tradeoff between the robustness level and the guaranteed utility. Consider a very restrictive social law (for example, autonomous vehicles are allowed to drive at a top speed of 10 miles per hour), which would have a high level of robustness ($\alpha$), but a low guaranteed utility ($E_i$). This social law can be contrasted with a less restrictive social law, which might be less robust but has a higher guaranteed utility.

Of course, each agent $i$ has a different guaranteed utility $E_i$ (and possibly a different robustness level), and would prefer to adopt a social law which maximizes its own $\alpha \cdot E_i$. Of course, it is up to the mechanism designer (or legislator) to choose a social law that is best for society overall, balancing social welfare, fairness, incentive compatibility, and other considerations. 
%
We leave such questions for future work and focus on the computational challenges of robustness verification. 







\section{Robustness Verification via Worst Case MDPs}

Having defined some computational questions about robustness, we now describe a technique to answer them.
Our technique is based on reducing the verification problem to the problem of solving MDPs.
In particular, we show that we can reduce the problem of computing the largest $\alpha$ for which a system is $\alpha$-robust to the problem of solving $|N|$ 'worst-case' MDPs that we construct using the original system and the optimal solutions to the projected MDPs.


Our method works in two stages. First, we solve the single agent projections optimally to obtain a compact characterization of each agent's optimal policies via the Bellman equation. Second, we construct for each agent a {\em worst-case} MDP, in which all agents are constrained to follow optimal policies of their single agent projections (obtained from the first step). An optimal solution of the worst-case MDP for agent $i$ corresponds to the worst combination of allowed policies for agent $i$. By taking the worst of the worst cases, we obtain the highest value of $\alpha$ for which the system is $\alpha$-robust. We now explain this method in more detail.

First, recall that the single agent projection for agent $i$, $\Pi_i$, is an MDP. Thus, its value function $\optValue{i}{s}$ is the solution to the following equation \citep{DBLP:books/wi/Puterman94}.
\begin{equation}
    \optValue{i}{s} = \max_{a \in A^i} R^i(s,a) + \gamma \sum_{s' \in S} P|_i(s,a,s')\optValue{i}{s'}.
\end{equation}
We use $\optAct{i}{s}$ to denote agent $i$'s optimal actions at state $s$, that is, the actions which attain the maximum value. 
\begin{equation}
    \optAct{i}{s} = \arg\max_{a \in A^i} R^i(s,a) + \gamma \sum_{s' \in S} P|_i(s,a,s')\optValue{i}{s'}.
\end{equation}

Because of the Markovian structure of the MDP, a stochastic policy $\pi$ is optimal in $\Pi_i$ if, for all states $s$, the policy only takes actions from $\optAct{i}{s}$ \citep{DBLP:books/wi/Puterman94}.
That is, $\sum_{a\in \optAct{i}{s}} \pi(s,a) = 1$ for all $s$ in $S$.

We make the following assumption to ensure optimal policies for agent $i$ necessarily take actions in $\optAct{i}{s}$ for every state $s$. We add this assumption to avoid pathological cases where agents use non-optimal actions at states that occur with zero probability under the optimal policy. We remark that this assumption is akin to a subgame-perfection condition \cite{fudenberg1983subgame}, under which agents' policies are optimal for any initial history.

\begin{assumption}
    \label{assump:qsa_reg}
    Agents use optimal single agent policies that take optimal actions at every state,
    that is, $\sum_{a\in \optAct{i}{s}} \pi^i(s,a) = 1$ for all $s$ in $S$.
\end{assumption}

Under Assumption~\ref{assump:qsa_reg}, a policy $\pi$ in MDP $\Pi_i$ is optimal if and only if, for all states $s$, the policy takes actions in $\optAct{i}{s}$. Thus, by optimally solving all the single agent projections, we can restrict the set of actions that we consider for agent $i$ in state $s$ to $\optAct{i}{s}$. Due to the Markovian nature of MDPs, the choice of action in state $s$ is independent of the choice of action in state $s'$ for any pair of states $s, s'$. Thus, $\optAct{i}{s}$ is a compact characterization of all optimal policies for agent $i$ in its single agent projection.

In the second step of this approach, we define a set of \emph{worst-case MDPs} to find the robustness value of $\Pi$. The worst-case MDP for agent $i$, denoted $\mathcal{M}_i$, does two things: (i) it restricts all agents to only choose optimal actions, based on the first step described above, and (ii) it assumes all agents act together, under centralized control, to find the worst possible set of joint policies for agent $i$. 

Formally, we define the worst-case MDP 
$\mathcal{M}_i = \langle S, \tilde{A}, I, -R^i, P, \gamma \rangle$ where $\tilde{A}$ is a state-dependent
action set.
The set of actions $\tilde{A}(s)$ available at state $s$ is
\begin{equation}
    \tilde{A}(s) = \times_{j=1}^{|N|} \optAct{j}{s} .
\end{equation}
Intuitively, a deterministic optimal policy in the worst-case MDP $\mathcal{M}_i$ prescribes a plan for the agents to achieve the worst possible outcome for agent $i$, which corresponds to maximally violating $\alpha$-robustness for agent $i$.
Furthermore, under this policy, agents take only actions that are optimal in their individual MDPs, and thus we can decompose it into individual plans for each agent, each optimal in its respective MDP.

\begin{theorem}
    \label{thm:transf_corr}
    Let $F_i = -\mathbb{E}_{s_0 \sim I}[V^*_{\mathcal{M}_i}(s_0)]$.
    Then $\alpha^* = \min_i F_i/E_i$ is the highest level of robustness achievable.
\end{theorem}

\textit{Proof:} 
Without loss of generality, suppose $\min_i F_i/E_i = F_1/E_1$.
We can not achieve a higher robustness than $\alpha^*$, as the deterministic optimal policy in worst-case MDP $\mathcal{M}_1$ 
provides a certificate that other agents can violate robustness even when following deterministic optimal policies.
However, $\alpha^*$ is achievable as a robustness level. Indeed, for any agent $i$, the product of any set of optimal policies 
$\{\pi_j^*\}_{j=1}^{|N|}$ is a valid policy in $\mathcal{M}_i$, as agent $j$'s policy must use actions in $\optAct{j}{s}$.
Under this policy, agent $i$ receives an expected reward of at least $F_i$, which is greater than $\alpha^* E_i$. $\square$

Theorem~\ref{thm:transf_corr} implicitly defines a simple algorithm for computing the robustness of a multi-agent system, which we detail in Algorithm~\ref{alg:alp_rob_dec}.

\begin{algorithm}[H]
\caption{Computing robustness of a multi-agent system}
\label{alg:alp_rob_dec}
\begin{algorithmic}[1]

\FOR{$i = 1$ to $|N|$}
    \STATE Compute value function $\optValue{i}{s}$
    \STATE Compute action set $\optAct{i}{s}$ for all states $s$
    \STATE Compute individual MDP value $E_i$
\ENDFOR

\STATE $\alpha^* \gets \infty$

\FOR{$i = 1$ to $|N|$}
    \STATE Construct worst-case MDP $\mathcal{M}_i$
    \STATE Solve worst-case MDP to obtain $F_i$
    \STATE $\alpha^* \gets \min(\alpha^*, F_i / E_i)$
\ENDFOR

\STATE \textbf{return} $\alpha^*$

\end{algorithmic}
\end{algorithm}

\subsubsection{Handling approximate value functions}

For certain approximate optimal value functions, we can provide lower bounds on the optimal robustness level $\alpha^*$.
In particular,
let $\underValueVec{i}$ and $\overValueVec{i}$ be lower and upper bounds on the optimal value for agent $i$, 
such that
\begin{equation}
    \forall i,s:\underValue{i}{s} \leq \optValue{i}{s} \leq \overValue{i}{s} .
\end{equation}
We can use these upper and lower bounds to compute a set $\optActOuter{i}{s}$ which overapproximates $\optAct{i}{s}$ as follows.
\begin{multline}
    \optActOuter{i}{s} = \{a \in A^i(s) | \\ R^i(s,a) + \gamma \sum_{s' \in S} P|_i(s,a,s')\overValue{i}{s'} \geq \underValue{i}{s} \}.
\end{multline}
If we compute the worst-case MDP $\overline{\mathcal{M}}_i$ using sets $\optActOuter{i}{s}$ to define $\tilde{A}(s)$, we can compute a lower bound $\underline{F}_i$ on $F_i$. 
We can also compute an upper bound $\overline{E}_i$ on $E_i$ using $\overValueVec{i}$. 
The value $\min_i{\underline{F}_i/\overline{E}_i}$ lower bounds $\alpha^*$.

Specifically, we can use the above approximation method to handle value functions which are close to $\optValueVec{i}$ in the $\infty$-norm, such as would be the case after running value iteration~\cite{Bellman:DynamicProgramming} for a finite number of steps.
Let $(\approxValueVec{i})_{i=1}^N$ be approximate value functions for the projected MDPs, 
that are $\tau$-close to $\optValueVec{i}$ in the $\infty$-norm.
That is, $|\approxValueVec{i}(s) - \optValue{i}{s} | \leq \tau$ for all $s \in S, i \in N$.
We can compute the above upper and lower bounds on $\optValueVec{i}$ by setting
\begin{equation}
    \overValueVec{i} = \approxValueVec{i} + \tau \mathbf{1} \text{ and } \underValueVec{i} = \approxValueVec{i} - \tau \mathbf{1}.
\end{equation}
In practice, this approximation is what we use with value iteration, since there may be minor numerical differences due to floating-point rounding and due to stopping after a finite number of iterations.




\section{Empirical Evaluation}

\begin{figure*}

\begin{tabular}{lp{0.3\linewidth}p{0.3\linewidth}p{0.3\linewidth}}
Env & $\alpha$ for Simple Grid & $\alpha$ for Grid w/ Velocity & Guarantee for Grid w/ Velocity\\
 \rotatebox{90}{~~~~~~~~~~~~~Parallel Lanes} &
\begin{tikzpicture}
\begin{axis}[
  xlabel=Length,
  ylabel=$\alpha$,
  legend pos=south west,
  height=4cm, width=\linewidth]
\addplot table [y=0.6, x=length]{data/aima/alpha_values_Parallel.dat};
\addplot table [y=0.8, x=length]{data/aima/alpha_values_Parallel.dat};
\addplot table [y=1.0, x=length]{data/aima/alpha_values_Parallel.dat};
\end{axis}
\end{tikzpicture}
&
\begin{tikzpicture}
\begin{axis}[
  xlabel=Length,
  ylabel=$\alpha$,
  legend pos=south west,
  height=4cm, width=\linewidth]
\addplot table [y=0.6, x=length]{data/alpha_values_Parallel_Speed.dat};
\addplot table [y=0.8, x=length]{data/alpha_values_Parallel_Speed.dat};
\addplot table [y=1.0, x=length]{data/alpha_values_Parallel_Speed.dat};
\end{axis}
\end{tikzpicture}
&
\begin{tikzpicture}
\begin{axis}[
  xlabel=Length,
  ylabel=Cost,
  legend pos=north west,
  height=4cm, width=\linewidth]
\addplot table [y=0.6, x=length]{data/w0_values_Parallel_Speed.dat};
\addplot table [y=0.8, x=length]{data/w0_values_Parallel_Speed.dat};
\addplot table [y=1.0, x=length]{data/w0_values_Parallel_Speed.dat};
\addplot table [y=1.0, x=length]{data/w0_values_Parallel_Speed_Slow.dat};
\end{axis}
\end{tikzpicture}\\
\rotatebox{90}{~~~~~~~~~~~~~Opposite Lanes} &
\begin{tikzpicture}
\begin{axis}[
  xlabel=Length,
  ylabel=$\alpha$,
  legend pos=outer north east,
  height=4cm, width=\linewidth]
\addplot table [y=0.6, x=length]{data/alpha_values_Opposite.dat};
\addplot table [y=0.8, x=length]{data/alpha_values_Opposite.dat};
\addplot table [y=1.0, x=length]{data/alpha_values_Opposite.dat};
\end{axis}
\end{tikzpicture}
&
\begin{tikzpicture}
\begin{axis}[
  xlabel=Length,
  ylabel=$\alpha$,
  legend pos=south west,
  height=4cm, width=\linewidth]
\addplot table [y=0.6, x=length]{data/alpha_values_Opposite_Speed.dat};
\addplot table [y=0.8, x=length]{data/alpha_values_Opposite_Speed.dat};
\addplot table [y=1.0, x=length]{data/alpha_values_Opposite_Speed.dat};
\end{axis}
\end{tikzpicture}
&
\begin{tikzpicture}
\begin{axis}[
  xlabel=Length,
  ylabel=Cost,
  legend pos=north east,
  height=4cm, width=\linewidth]
\addplot table [y=0.6, x=length]{data/w0_values_Opposite_Speed.dat};
\addplot table [y=0.8, x=length]{data/w0_values_Opposite_Speed.dat};
\addplot table [y=1.0, x=length]{data/w0_values_Opposite_Speed.dat};
\addplot table [y=1.0, x=length]{data/w0_values_Opposite_Speed_Slow.dat};
\end{axis}
\end{tikzpicture}\\
\rotatebox{90}{~~~~~~~~~~~~~Switch} &
\begin{tikzpicture}
\begin{axis}[
  xlabel=Length,
  ylabel=$\alpha$,
  legend pos=north west,
  height=4cm, width=\linewidth, ymin=0,ymax=1]
\addplot table [y=0.6, x=length]{data/alpha_values_CW.dat};
\addplot table [y=0.8, x=length]{data/alpha_values_CW.dat};
\addplot table [y=1.0, x=length]{data/alpha_values_CW.dat};
\end{axis}
\end{tikzpicture}
&
\begin{tikzpicture}
\begin{axis}[
  xlabel=Length,
  ylabel=$\alpha$,
  legend pos=north west,
  height=4cm, width=\linewidth, ymin=0,ymax=1]
\addplot table [y=0.6, x=length]{data/alpha_values_CW_Speed.dat};
\addplot table [y=0.8, x=length]{data/alpha_values_CW_Speed.dat};
\addplot table [y=1.0, x=length]{data/alpha_values_CW_Speed.dat};
\end{axis}
\end{tikzpicture}
&
\begin{tikzpicture}
\begin{axis}[
  xlabel=Length,
  ylabel=Cost,
  legend pos=north east,
  height=4cm, width=\linewidth]
\addplot table [y=0.6, x=length]{data/w0_values_CW_Speed.dat};
\addplot table [y=0.8, x=length]{data/w0_values_CW_Speed.dat};
\addplot table [y=1.0, x=length]{data/w0_values_CW_Speed.dat};
\addplot table [y=1.0, x=length]{data/w0_values_CW_Speed_Slow.dat};
\end{axis}
\end{tikzpicture}\\
\rotatebox{90}{~~~~~~~~~~~~~Switch (CW SL)} &
\begin{tikzpicture}
\begin{axis}[
  xlabel=Length,
  ylabel=$\alpha$,
  legend pos=south east,
  height=4cm, width=\linewidth]
\addplot table [y=0.6, x=length]{data/alpha_values_CWSL.dat};
\addplot table [y=0.8, x=length]{data/alpha_values_CWSL.dat};
\addplot table [y=1.0, x=length]{data/alpha_values_CWSL.dat};
\end{axis}
\end{tikzpicture}
&
\begin{tikzpicture}
\begin{axis}[
  xlabel=Length,
  ylabel=$\alpha$,
  legend pos=south east,
  height=4cm, width=\linewidth]
\addplot table [y=0.6, x=length]{data/alpha_values_CWSL_Speed.dat};
\addplot table [y=0.8, x=length]{data/alpha_values_CWSL_Speed.dat};
\addplot table [y=1.0, x=length]{data/alpha_values_CWSL_Speed.dat};
\end{axis}
\end{tikzpicture}
&
\begin{tikzpicture}
\begin{axis}[
  xlabel=Length,
  ylabel=Cost,
  height=4cm, width=\linewidth]
\addplot table [y=0.6, x=length]{data/w0_values_CWSL_Speed.dat};
\addplot table [y=0.8, x=length]{data/w0_values_CWSL_Speed.dat};
\addplot table [y=1.0, x=length]{data/w0_values_CWSL_Speed.dat};
\addplot table [y=1.0, x=length]{data/w0_values_CWSL_Speed_Slow.dat};
\end{axis}
\end{tikzpicture}\\
& \multicolumn{3}{c}{
\begin{tikzpicture}
\begin{axis}[
  legend style={at={(0,1.0)}, anchor=north,legend columns=-1,/tikz/every even column/.append style={column sep=0.5cm}},
  axis lines = none,
  height=2cm, width=2cm]
\addplot{0};
\addlegendentry{$p_{\mbox{success}}=0.6$}
\addplot{0};
\addlegendentry{$p_{\mbox{success}}=0.8$}
\addplot{0};
\addlegendentry{$p_{\mbox{success}}=1.0$}
\addplot{0};
\addlegendentry{slow}
\end{axis}
\end{tikzpicture}
}
\\
\end{tabular}

    \caption{\label{f:lanes} Results for Different Grid Environments for Different Length and $p_{\mbox{success}}$ Values.\\ 
    The left column shows the robustness values for the simple grid environment.
    The middle column shows the robustness values for the grid environment with velocities.
    The right column shows the guaranteed utility (cost) for the grid environment with velocities, without social law, and under the social law that forces all agents to go slowly.
    }
\end{figure*}
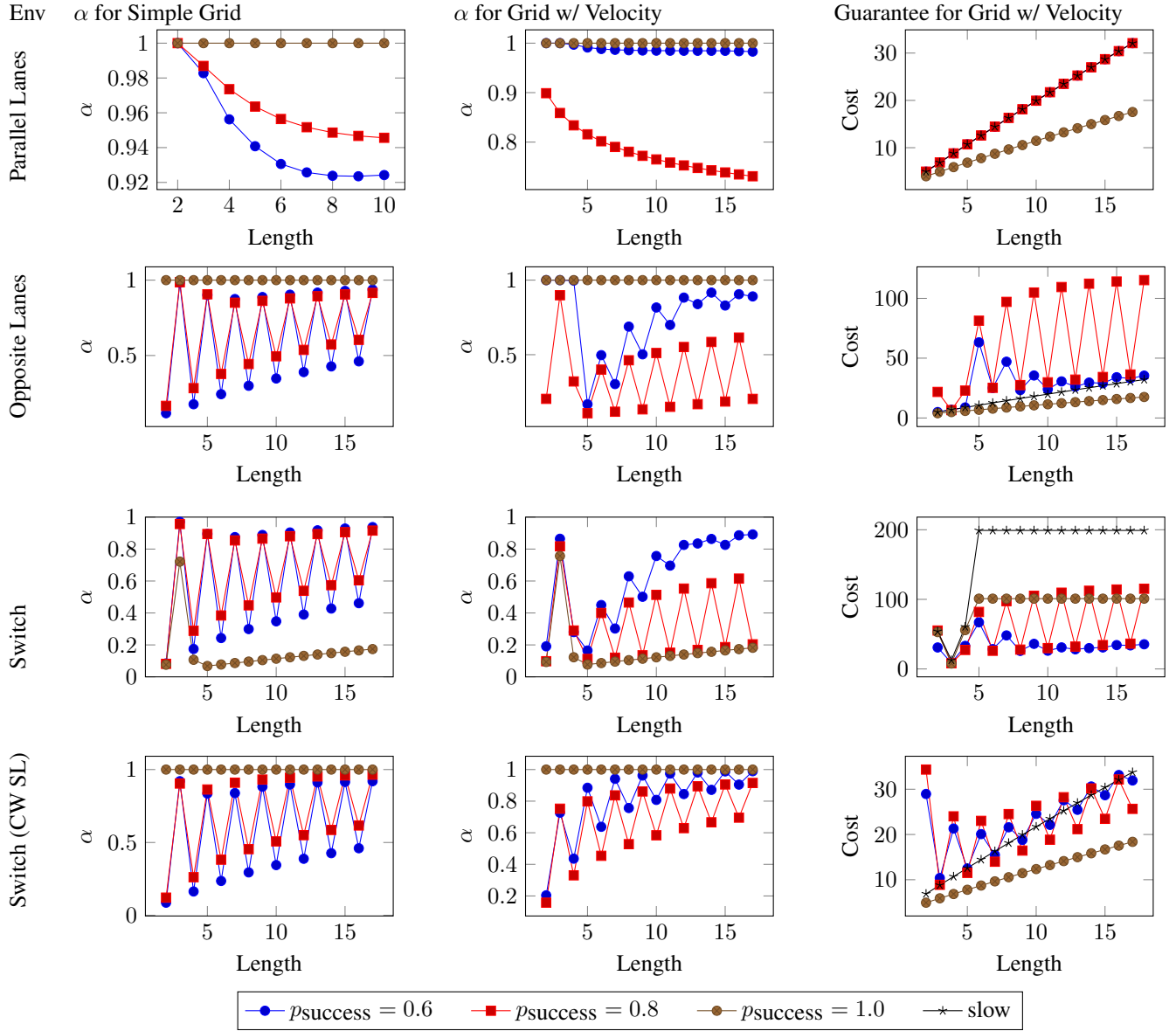

To check the performance of our proposed approach, we performed an empirical evaluation.
We implemented our approach in Python, relying on the basic MDP implementation from \cite{DBLP:books/aw/RN2020}. We used value iteration to solve the MDPs on several benchmarks.

\subsection{Simple Grid}

\paragraph{Description:}First, we used a simple grid in which 2 agents move, which is a generalization of the example in Figure \ref{f:example}. Each agent starts in an $\textit{outside}$ state, and has a ``drop in'' action which places the agent in a designated start position.
This ``drop in'' allows the grid in the single agent projection to be completely clear in the initial state, allowing the agent to choose any path it wants. 
Each agent can move east/west/north/south, with probability $p_{\mbox{success}}$ of reaching the desired position. Otherwise, the agent will end up in one of the other adjacent cells. Finally, upon reaching some designated goal location, the agent can execute a ``leave'' action, which moves it to the absorbing state $\textit{done}$ in which it collects a reward of 0. 
Additionally, if two agents end up in the same location, they experience a collision, which is an absorbing state with a reward of -100. All other states have a reward of -1. We used a discount factor $\gamma=0.99$.

In this generic grid problem, we developed several specific benchmarks, all of them on grids of size 2$\times n$, where we can vary $n$:

\begin{description}
    \item[Parallel Lanes:]  One agent is trying to get from the northwest to the northeast, and the other from southwest to southeast.  
    \item[Opposite Lanes:] One agent is trying to get from the northwest to the northeast, and the other from southeast to southwest.
    \item[Switch Corners:] One agent is trying to get from northwest to southeast, and the other from southeast to northwest.
    \item[Switch Corners (CW SL):] This is the same switch corners scenario, but with a social law (SL) which forces agents to move only clockwise (CW), and move north/south at the latest possible point (that is, the first agent can only go south at the eastern edge of the grid, and the second agent can only go north at the western edge of the grid).
\end{description}

\paragraph{Experiment Results:} 
The leftmost column of Figure \ref{f:lanes} shows the results for a simple grid, for various lengths $n$, and various values $p_{\mbox{success}}$, for the four scenarios above.
In the parallel-lane and opposite-lane environments, under the $p_{\mbox{success}} = 1$ regime, each agent has a single optimal policy to drive in its lane. Thus, as evident in the results, we get a robustness level of $\alpha=1$. 
In the Switch Corners environment, each agent has many possible optimal policies (since it must go north or south at some point, and all choices are equally valid). Thus, we get low robustness values.
However, for the Switch Corner with the social law, we regain $\alpha=1$ robustness when $p_{\mbox{success}} = 1$.
Of course, in all scenarios when $p_{\mbox{success}} < 1$, it becomes impossible to guarantee perfect robustness, as is evident in the results.

\subsection{Grid With Velocity}
\paragraph{Description:} We define another benchmark by modifying the simple grid benchmark. We extend the action space of the simple grid benchmark by allowing each agent to move either fast or slow, yielding eight actions that let them move in distinct cardinal directions at different velocities. When moving slowly, agents always successfully move to the adjacent grid cell in the selected cardinal direction. In contrast, agents may slip when deciding to move quickly, ending up in an adjacent grid that is not the intended destination of the chosen action, with probability $1 - p_{\mbox{success}}$. The remaining actions and the distribution of next states from selecting them are the same as the simple grid environment.

In designing this new benchmark, we also modify the reward function of the simple grid benchmark. Agents are specifically given a reward of -1 whenever they choose to move fast. Meanwhile, reaching any state without selecting any actions to move quickly yields a reward of -2. The only exceptions are when agents choose the "leave" action at the goal location, which yields a reward of 0, and when they collide, which yields a reward of -100. Finally, we also use a discount factor of $\gamma = 0.99$ in our experiments with this benchmark.

\paragraph{Experiment Results:}

The middle column of Figure \ref{f:lanes} shows the robustness values of this benchmark, for the same four scenarios as in simple grid, varying the length of the grid $n$ as well as $p_{\mbox{success}}$ (which now only applies to moving fast).
Unsurprisingly, the robustness values look similar to the simple grid, specifically with $\alpha=1$ for parallel, opposite, and switch corners with the CW social law when $p_{\mbox{success}}=1$. 

More interestingly, this benchmark allows us to examine another social law: the requirement that agents always move slowly. In terms of robustness, this is equivalent to moving to the $p_{\mbox{success}}=1$ regime, so we get a robustness value of 1 for parallel, opposite, and switch corners with CW social law. However, we can now explore the impact of adopting this social law on the guaranteed utility each agent can achieve, $\alpha \cdot E_i$. Since the agents are symmetric here, we consider only one of them.
The right column of Figure \ref{f:lanes} shows the guaranteed utility (in terms of cost --- lower is better) under the {\em slow} social law, for the four scenarios under grid with velocity. Note that for the bottom right plot, we apply both the  CW social law and the slow social law --- so agents are only allowed to move at the slow speed in a clockwise direction.


As these results show, going slow is not always optimal. In all scenarios, when $p_{\mbox{success}}=1.0$, going slow is a useless restriction, as it does not improve robustness. Thus, we do not expect to do better than the $p_{\mbox{success}}=1.0$, but are interested in comparing going slow to not going slow when  $p_{\mbox{success}}=0.6$ and $p_{\mbox{success}}=0.8$.
For parallel lanes, going slow yields the exact same cost with $p_{\mbox{success}}=0.6$ and $p_{\mbox{success}}=0.8$, and the lines align.
On the other hand, for the opposite-lane case, going slow yields higher utility than not going slow when $p_{\mbox{success}}=0.6$ and $p_{\mbox{success}}=0.8$.  
For switching corners under the CW social law, going slowly is better for small grids but worse for larger grids. Finally, for switching corners without the social law, going slowly does not improve robustness and thus unnecessarily increases the cost.

As these results show, and as is consistent with our theoretical results, not every social law that guarantees robustness is always a good idea. Specifically for the grid, it might be worthwhile to consider more complex social laws, such as slowing down when another agent is nearby.



\section{Conclusion and Future Work}

We have presented a new framework for multi-agent coordination in stochastic environments based on social laws. We defined the notion of $\alpha$-robustness, and showed how it can be used with the guaranteed utility to provide guarantees on what each agent can achieve in a multi-agent setting, while planning only for itself. We have also showed examples where using social laws increases the utility each agent can guarantee, and other examples where this is not the case.

Unlike previous work on planning-based social laws in deterministic settings
\cite{Karpas2017}, where the "waitfor" construct was necessary to prevent possible action failures (at the price of introducing potential deadlocks), using stochastic games as the underlying formalism makes this construct irrelevant.

This is because agents execute actions concurrently without waiting for specific states to be achieved beforehand. In fact, in our setting, waitfors can be implemented as a policy which executes a noop action in some states, but there is no need to handle waitfors via a special mechanism. This eliminates some of the rigidity introduced by deterministic environments and creates new opportunities for designing coordination mechanisms.

The social law verification problem we study requires knowledge of all agents' optimal value functions, or at least an approximately optimal value function for each agent. It is a simple exercise to extend these definitions to allow each agent to have a suboptimality level $\beta$. The actions we consider for the worst-case MDP would be those that are $\beta$-suboptimal or better.

Another interesting direction for future research is to use reinforcement learning (RL) to solve the MDPs at hand. While RL does not guarantee optimality, we have theoretical results about approximate value functions. Another challenge is that RL typically finds {\em one} optimal policy, while our technique requires the optimal value function. One option would be to use 1-step value updates to compute the optimal value function based on an approximately optimal policy from RL. Another option is to use RL techniques with a critic, and use the critic for the value function.

Finally, while we focused on robustness verification in this paper, we would like to also be able to synthesize robust social laws automatically. In deterministic settings, this has been done via searching through the sets of disallowed actions \cite{DBLP:conf/aaai/NirSK20}, and we intend to follow a similar approach for stochastic settings. In deterministic settings, it was possible to prune this search space by considering only actions which appeared in a counter-example to robustness. While this is also possible in stochastic settings, we will need to consider all actions that appear in the optimal policies that solve the worst-case MDPs, even those that are rarely executed. Thus, we intend to develop stronger pruning rules for the search for robust social laws.




\section{Acknowledgments}

This work has taken place in the Learning Agents Research Group (LARG) at UT Austin.  LARG research is supported in part by NSF (FAIN-2019844, NRT-2125858), ONR (N00014-24-1-2550), ARO (W911NF-17-2-0181, W911NF-23-2-0004, W911NF-25-1-0065), DARPA (Cooperative Agreement HR00112520004 on Ad Hoc Teamwork) Lockheed Martin, and UT Austin's Good Systems grand challenge. Peter Stone serves as the Chief Scientist of Sony AI and receives financial compensation for that role.  The terms of this arrangement have been reviewed and approved by the University of Texas at Austin in accordance with its policy on objectivity in research.

Caleb Probine and Ufuk Topcu were funded by the Office of Naval Research under grant number ONR N00014-25-1-2479 and the Air Force Office of Scientific Research under grant number AFOSR FA9550-22-1-0403.

\bibliography{refs}


\end{document}